\RequirePackage{fix-cm}
\documentclass[smallextended]{svjour3}       
\smartqed  
\usepackage{graphicx}
\usepackage{mathptmx}      
\usepackage{amsmath,amssymb}
\usepackage[english]{babel}
\usepackage{hyperref}
\usepackage{tikz}
\newcommand{\cB}{\mathcal{B}}
\newcommand{\cM}{\mathcal{M}}
\newcommand{\cH}{\mathcal{H}}
\newcommand{\e}[1][q]{\textbf{e}_{#1}}
\newcommand{\Discrepancy}[1]{D_{N}(#1)}
\newcommand{\F}[1][q]{\mathbb{F}_{#1}}
\newcommand{\ck}{\tilde{k}}
\newcommand{\cL}{\mathcal{L}}
\renewcommand{\vec}[1]{\mathbf{#1}}
\newcommand{\cT}{\mathcal{T}}
\newcommand{\vol}{\mathrm{vol}}
\newcommand{\cK}{\mathcal{K}}
\renewcommand{\(}{\left (}
\renewcommand{\)}{\right )}
\newtheorem{assumption}{Assumption}
\begin{document}

\title{A Probabilistic
  Analysis on a Lattice Attack against DSA\thanks{Domingo Gomez-Perez was partially supported by project MTM2014-55421-P from the Ministerio de Economia y Competitividad.}
}


\author{Domingo Gomez-Perez$^{a}$   \and Gu{\'e}na{\"e}l Renault$^{b,c}$}


\institute{
  a: Avd. Los Castros s/n\\
  Facultad de Ciencias\\
  Universidad de Cantabria\\
  Santander, 39005, Spain.\\
  b: Agence Nationale de la S\'ecurit\'e des Syst\`emes d'Information\\
  51, boulevard de La Tour Maubourg\\
  75700 Paris 07 SP, France.\\
  c: Universit\'e Pierre et Marie Curie\\
  LIP6 - \'Equipe projet INRIA/UPMC POLSYS\\
  4 place Jussieu\\
  F-75252 Paris cedex 5 France\\
}

\date{Received: date / Accepted: date}

\maketitle

\begin{abstract}
  Analyzing the security of cryptosystems under attacks based on the malicious
  modification of memory registers is a research topic of high importance. This
  type of attacks may affect the randomness of the secret parameters by forcing
  a limited number of bits to a certain value which can be unknown to the
  attacker. In this context, we revisit the attack on DSA presented by Faugè\`ere,
  Goyet and Renault during the conference SAC 2012: we simplify their method and
  we provide a probabilistic approach in opposition to the heuristic proposed in
  the former to measure the limits of the attack. More precisely, the main
  problem is formulated as the search for a closest vector to a lattice, then we
  study the distribution of the vectors with bounded norms in a this family of
  lattices and we apply the result to predict the behavior of the attack. We
  validated this approach by computational experiments.

\end{abstract}

\section{Introduction}

The security of the main public-key cryptosystems is based on the difficulty of
solving certain mathematical problems. In this context, the most commonly used
problems come from Number Theory, most notably the integer factorization problem
and the discrete logarithm on finite cyclic groups. DSA and RSA are two of the
most used cryptosystems and their security relays in these problems. There are
many software products like SSH, OpenPGP, S/MIME and SSL which use RSA for
encrypting and signing and DSA for signing. The National Institute of Standards
and Technology has also promoted the use of elliptic curve cryptography, whose
security is based on the discrete logarithm problem in special groups. Although,
the National Security Agency has recently advocated to start replacing these
cryptosystems~\cite{NSA} because of the potential developments in quantum
computing, the perspective is that these cryptosystems are going to be widely
use in the short term.

In this paper, we study a cryptosystem based on the discrete logarithm problem.
Apart from the advances in quantum computing and some recent results on
quasi-polynomial complexity algorithms for solving the discrete logarithm
problem in multiplicative groups of small characteristic
fields~\cite{Jouxetal,Grangeretal1,Grangeretal2,Joux13} and for certain abelian
groups~\cite{sutherland2011structure}, the discrete logarithm problem on finite
fields of large characteristic remains solvable in subexponential time only and
the best algorithm known is given by Adleman and DeMarrais~\cite{DLPAdlemanD93}.

The discrete logarithm problem in elliptic curves seems to be even
harder. Although there are results
for anomalous curves~\cite{cohen2005handbook} and
curves defined in extension fields~\cite{diem2013discrete},
known approaches run in exponential time
(see the survey by Galbraith and Gaudry~\cite{galbraith2015}).

However, this fact does not mean that attacking secure cryptosystems is
hopeless. Many practical attacks are possible because there is additional
information available due to the knowledge of implementation. For example,
Genkin, Shamir and Tromer~\cite{accustic} showed that it is possible to recover
the private key of a 4096-RSA cryptosystem using the sound pattern generated
during the decryption of some chosen data.

These advances point to a new research question: which information should be
added in order to solve this problem in polynomial time. This question has been
in the spotlight for a long time. Indeed, Rivest and
Shamir~\cite{RSAfactoring86} introduced the notion of oracle to formalize this
approach in the context of factorization of RSA modules.

In this article, we focus on the Digital Signature Algorithm (DSA) \cite{DSS94}
whose security is based on the difficulty of the DLP in multiplicative groups of
finite fields (see Section \ref{Def} for more details). The first proposal of
using an oracle on DSA comes from Howgrave-Graham and
Smart~\cite{Howgrave-GrahamS01} using the LLL lattice reduction
algorithm~\cite{LLLfactoring} to take benefit from the knowledge of a small
number of bits in many ephemeral keys. However, these results were only
heuristics, even though confirmed by experimentation. Nguyen and Shparlinski~
\cite{NguyenShparlinski02,NguyenShparlinski03} presented the first polynomial
time algorithm that provably recovers the secret DSA key if about
$\log^{1/2}(q)$ LSB (or MSB) of each ephemeral key are known ($q$ denoting the
order of the chosen group, see Section~\ref{Def}) for a polynomially bounded
number of corresponding signed messages. Other attacks take advantage using the
bits in the ephemeral key and the Fast Fourier
Transform~\cite{de2013using,de2014}. We remark that, although, these
type of attacks normally need less bits, the computational cost is bigger.
However, there is a common point between these attacks. They need
\textit{explicit information} about the bits used and they bypass the problem of
computing discrete logarithms.

At SAC 2012, Faug\`ere, Goyet and Renault~\cite{Faugere12} restricted the power
of the oracle by introducing an implicit attack on DSA. More precisely, they do
not assume that the oracle explicitly outputs bits of the ephemeral keys but
rather provides only \textit{implicit information}. In this implicit scenario,
the oracle is stated in the following way: the attacker knows some signatures
that were computed with ephemeral keys sharing some bits. Instead of an explicit
information related to the value of these shared bits the implicit information
provides only the positions of the shared bits. In an application point of view,
this oracle can be instantiated by an invasive attack where some registers used
by a pseudo random generator would be destroyed by a laser and keep always the
same unknown value during the computation of many signatures. The introduction
of implicit information given by an oracle where first presented by May and
Ritzenhofen~\cite{MayR09} in the context of the RSA cryptosystem and well
studied since then (e.g. \cite{FMR10,sarkar2009further}). The attack proposed in
\cite{Faugere12} is heuristic based. The contribution of this article is to
provide a rigorous proof and analyze the applicability of this attack. This
article presents results for the DSA over a finite field, but we remark that
these techniques can be adapted for the elliptic curve version (ECDSA) as well.

The paper is organized as follows. Section~\ref{Def} gives an overview of DSA,
recalls the attack proposed in~\cite{Faugere12} and presents the main
contribution of this paper. Section~\ref{sec:short_and_discrepancy} presents the
background in uniform distribution theory necessary to understand the
probabilistic approach. Section~\ref{sec:exponential} presents the proofs of our
main results and Section~\ref{sec:experiments} shows the performance of the
attack in experiments and discusses the relation with the theoretical results.

\section{Implicit attack on DSA}\label{Def}
We follow the same notation as in the article~\cite{Faugere12} and go through
the technique proposed there. The next diagram represents the protocol to
generate a public key and signing a message using DSA with finite fields. For
readers not familiar with DSA, we provide the explicit details. \\ \*
 \begin{center}
   \begin{tikzpicture}
     \node [above] at (-1,3.2) {step 0};
     \node [above] at (3,3.2) {step 1};
     \node [above] at (6,3.2) {step 2};
     \draw [thick]  (-1,3) -- (9,3);
     \draw (-1,2.8) -- (-1, 3.2);
     \draw (3,2.8) -- (3, 3.2);
     \draw (6,2.8) -- (6, 3.2);
     \draw (9,2.8) -- (9, 3.2);
     \node[align=left, below] at (1,2.5)%
     {User generates\\ $p$, $q|(p-1)$\\
       $g\in\F[p]$ of order $q$\\
      and key $1\le a\le q.$};
     \node[align=center, below] at (4.5,2.5)%
     {User publishes\\$g,p,q,g^a.$};
     \node[align=right, below] at (7.5,2.5)%
     {To sign $m$\\ user calculate $r,s$\\ as in Equation~\eqref{definition:key}\\
       and send $m,r,s$.};
   \end{tikzpicture}
 \end{center}
Let $M$ be a positive integer, $p$ be a a $L$-bit prime and $q$ be a prime
divisor of $p-1$ satisfying $2^{M-1}<q< 2^{M}$. The integers $p$ and $q$ are
recommended to be chosen such that $(p,q) \in \{(1024,160), (2048,224),(2048,256),(3072,256)\}$ see \cite{FIPS:2013:DSS}.

The finite field of $q$ elements is denoted by $\F[q]$ and each of its element
is uniquely represented by an integer in the range $\{(1-q)/2,\ldots,
(q-1)/2\}$. This also implies that in the sequel, any number modulo $q$ gives a
number in the previous range. For the DSA signature scheme, the user selects a
random element $a\in\F[q]$, which must be kept private, and then publishes $q$,
$p$, an element $g\in\F[p]$ of multiplicative order $q$ and $g^{a}\bmod p$.

For efficiency and security reasons, the bit-size of the messages signed with
DSA has to be the same as the one of $q$ (e.g. $160$ or $256$). Thus for a
general message it is necessary to consider its hash and only sign this hash. In
the sequel, we denote by $\cH$ this hash function and the hash of the message
with $m$ (which its bit-size is assumed to be adapted to the chosen $q$). The
hash function is not important in the results, if it has standard security
requirements.

To sign $m$, the user generates a random number $k\in\F[q]$ (called the
ephemeral key) and calculates,
\begin{equation}
\label{eq:definition_rs}
 r:=g^{k}\mod p\mod q\quad\text{ and }\quad s:=k^{-1}(m+ar)\bmod q.
\end{equation}
The user requires that $r$ and $s$ are not zero, and in this case, $(r,s)$ is a
valid signature. Otherwise, the user generates another $k$ and calculates
$(r,s)$ again.

\subsection{Scenario of the attack}

We suppose that the user wants to sign $n$ messages, whose hashes are
$m_1,\ldots, m_n,$ so he generates $k_1,\ldots,k_n$ and publishes the signatures
$(s_1,r_1),\ldots,$ $(s_n,r_n)$. We also suppose that, due to some malicious
actions of the attacker, the corresponding ephemeral keys differs only in a
block of bits of known length so the attacker knows that,
\begin{equation}
\label{definition:system}
  \begin{split}
    m_1+ar_1-s_1k_1&=0\bmod q,\\
    m_2+ar_2-s_2k_2&=0\bmod q,\\
    \vdots\ \ \ \ & \\
    m_n+ar_n-s_nk_n&=0\bmod q,\\
  \end{split}
\end{equation}
where $k_i$ have the following property,
\begin{equation}
\label{definition:key}
  k_i=k'+2^{t}\tilde{k_i}+2^{t'}k'',\quad
  |\tilde{k}_i|\le 2^{M-\delta},
  \quad\text{ for }i=1,\ldots,n
\end{equation}
with $k'$, $k''$ two unknown fixed $t$-bit and $(M-t')$-bit integers
respectively. Thus, there is a total of $\delta=M-t'+t$ shared bits. Notice that
we can substitute $k_i$ using Equation~\eqref{definition:key} in
Equation~\eqref{definition:system} and eliminate variables $k'$ and $k''$ which
results in the following set of equations,
\begin{equation*}
  \begin{split}
    a\beta_2&=\alpha_2+\ck_1 -\ck_2\bmod q,\\
    a\beta_3&=\alpha_3+\ck_1-\ck_3\bmod q,\\
    & \vdots \\
    a\beta_n&=\alpha_n+\ck_1-\ck_n\bmod q,\\
  \end{split}
\end{equation*}
where $\ck_i$ come from~\eqref{definition:key} and
$\alpha_i,\ \beta_i\in \F$ are public values defined as,
\begin{equation}
\label{eq:linear_equations}
  \begin{split}
    \alpha_i&:=2^{-t}(s_i^{-1}m_i-s_1^{-1}m_1)\bmod q,\\
    \beta_i&:=2^{-t}(s_1^{-1}r_1-s_i^{-1}r_i)\bmod q.\\
  \end{split}
\end{equation}
Next, we can build a lattice using the rows of the following matrix,
\begin{equation}
\label{eq:lattice_faugere}
  \(
  \begin{matrix}
   2^{M-\delta}&0& \alpha_2&\alpha_3&\ldots&\alpha_n\\
   0&2^{-\delta}& \beta_2&\beta_3&\ldots&\beta_n\\
   0&0&q&0&\ldots&0\\
   0&0&0&q&\ldots&0\\
   \vdots&\vdots&\vdots&\vdots&\vdots&\vdots\\
   0&0&0&0&\ldots&q\\
  \end{matrix}
  \),
\end{equation}
and find a short vector in it using an appropriate algorithm,
for example~\cite{schnorr}.
The attacker hopes to recover the following vector,
\begin{equation*}
  (2^{M-\delta},-a2^{-\delta},\ck_2-\ck_1,\ldots, \ck_n-\ck_1),
\end{equation*}
which has a rather short norm.
This is the algorithm proposed in~\cite{Faugere12},
with some discussion depending on the parameters.

\subsection{Contributions}

Our first contribution is a variant of this proposal, we still relate the
recovering of the ephemeral keys in DSA with a lattice problem but we give
rigorous results on the performance of the resulting algorithm.

To give this new attack, we follow the presentation in~\cite{Nguyen02}.
First, we define the lattice $\cL$  by the rows of the following
matrix,
\begin{equation}
  \label{eq:matrix}
  \(
  \begin{matrix}
    2^{-\delta}&\beta_2&\beta_3&\ldots&\beta_n\\
    0&q&\ldots&0&0\\
   \vdots&\vdots&\vdots&\vdots&\vdots\\
    0&0&\ldots&q&0\\
    0&0&\ldots&0&q\\
  \end{matrix}
  \),
\end{equation}
and two vectors $\vec{t},\ \vec{u}$,
\begin{equation}
  \label{eq:twovectors}
  \begin{split}
  \vec{t}&=(0,\alpha_2,\alpha_3,\ldots, \alpha_n),\\
  \vec{u}&=r(a2^{-\delta},\gamma_2,\gamma_3,\ldots, \gamma_n),\\
  \end{split}
\end{equation}
where $\gamma_i:=a\beta_i\bmod q$ for $i=2,\ldots, n$.

Lattice $\cL$ and $\vec{t}$ are known to the attacker and his goal is to recover
$\vec{u}$ using this information. It is straightforward that $\vec{u}\in\cL$ and
$\|\vec{u}-\vec{t}\|\le \sqrt{n}2^{(M-\delta)}$. Thus $\vec{u}$ is a vector in
this lattice, which is close to $\vec{t}$, and we hope that the solution of the
closest vector problem is $\vec{u}$.

If we call $\vec{h}\in\cL$ the solution to the closest vector problem then
$\vec{v}=\vec{u}-\vec{h}$ verifies $\vec{v}\in\cL$ and
$\|\vec{v}\|=\|\vec{u}-\vec{h}\|\le \|\vec{u}-\vec{t}\|+\|\vec{t}-\vec{h}\|\le
\sqrt{n}2^{M-\delta+1}$.

The so-called Gaussian heuristic (see~\cite[page 27, Definition
  8]{nguyen2010lll}) provides a way of analyzing this method's
performance, describing those cases where $\vec{h}$ is expected to be
$\vec{u}$. The shortest vector of lattice $\cL$ is expected to have
norm
\begin{equation*}
  \sqrt{\frac{n+1}{e\pi}}\vol(\cL)^{1/n}\approx
  \sqrt{\frac{n+1}{e\pi}} q^{1-1/n}2^{-\delta/n}
\end{equation*}
so, as soon as
\begin{equation*}
   q^{2-2/(n+1)}2^{-\delta/n}\ge  (e\pi) 2^{2(M-\delta+1)},
\end{equation*}
we hope to recover $\vec{u}$.

Applying the Gaussian heuristic to the lattice defined in
Equation~\eqref{eq:lattice_faugere} is equivalent to this situation
because a short vector $\vec{u}\in\cL$ defines the
short vector $(0,\vec{u})$ which is in the lattice defined
in~\eqref{eq:lattice_faugere}.
This argument is heuristic in nature, so an attacker
who finds  the closest vector in $\cL$ to $\vec{t}$ has no theoretical
guarantee to rediscover $a$. We extend this argument to a
probabilistic-in-nature argument.  This means, we can measure the
success probability of this attack.

\paragraph{Assumption and Statement of the Main Result.}

In order to state our main result, we need that the hash function $\cH$ used in DSA
verifies a property (which is the case in practice):

\begin{assumption}\label{ref:assumption}
  Let $\mathbb{M}$ and $\mathbb{M}'$ be two different messages. The probability
  of a collision
  \[
  \cH(\mathbb{M})=\cH(\mathbb{M}')
  \]
  is supposed to be less than $q^{-d-1}$ where $d$ is some positive constant
  that will be defined later.
\end{assumption}



Under this assumption, we can now state our main result.
\begin{theorem}
  \label{theorem:main}
  Under the notations used above and Assumption~\ref{ref:assumption}, there
  exists $d>0$ such that the probability that $\vec{u}$ is the solution of the
  closest vector problem where $\vec{t}$ is the target vector in $\cL$ is
  greater than $1-q(2^{-\delta+\log n +1}+ q^{-d})^{n-1}$.
\end{theorem}
This is equivalent to say that when $M < (n-1)\min\{\delta-\log n -1, d M\}$, then
the attack has non negligible probability of being successful.

We remark that this assumption is not a big restriction because the expected
probability of collision on a good hash function is of order $q^{-1}$.

Also, although $d$ is difficult to be evaluated exactly, if $q \ge
p^{\varepsilon}$ for some positive $0<\varepsilon<1/2$, then $d\ge
2^{145-82/\varepsilon}-\log\log p/\log q$ when $p$ is sufficiently big. This
gives a lower bound for the success probability of the algorithm.

Moreover, we conjecture that the value of $d$ is close to $0.5$, so
if $\delta\ ge \log^{1/2}(M)$ and $ n\ge M/\log^{-1/2} (M),$ the probability of
success is greater than $1- 1/n$.

This theorem can be generalized in the case where each ephemeral key is taken
with $\ell$ blocks of bits fixed sharing a total of $\delta$ bits. This case was
also considered in \cite{Faugere12} with a heuristic approach. Again, in order
to obtain a probability of success of the attack, the hash function has to
verify Assumption~\ref{ref:assumption}.

\begin{theorem}
  \label{thm:general}
  Under the notation used above, and generalizing the attack above for
  $k_1,\ldots, k_n$ having $\ell$ blocks of bits sharing a total of $\delta$ bits,
  there exists $d>0$ and a
  probabilistic  algorithm to recover $a$ in polynomial
  time in the   size of the input such that the success probability
  is greater than $1-q(2^{-\delta+\log n+1}+ q^{-d}(\log q)^{\ell})^{n-1}$.
\end{theorem}
For practical purposes, the most interesting case is $\ell=1$, so we focus on
this case, the proof of the general case follows the same ideas with more
technicalities.

\section{Short vectors and Discrepancy measures}
\label{sec:short_and_discrepancy}
Coming back at our original problem, we remark that we want to prove that the
solution of the CVP is, in some way unique and this is related with the norm of
the shortest vector in the lattice $\cL$. This lattice has a vector of norm at
most $2^{M-\delta}$ if and only if there exists $b\in\F[q]$, such that
\begin{equation*}
  b\beta_i = h_i\bmod q,\quad\text{ where }|h_i|\le 2^{M-\delta},\
  i=2,\ldots, n.
\end{equation*}
If $\beta_i$ were taken randomly and independently in $\F[q]$, then the
probability of this event is approximately $q 2^{-\delta(n-1)}$. More precisely,
we have the following result from \cite{Nguyen02}.
\begin{lemma}[\cite{Nguyen02}]
  Let $a\in\F$ be different from zero. Choose integers $\beta_2,\ldots, \beta_n$
  uniformly and independently at random in $\F$. Then with probability $P\ge 1-
  q2^{-\delta(n-1)}$ all vectors $\vec{v}\in\cL$ such that
  $\|\vec{v}\|_{\infty}\le 2^{M-\delta}$ are of the form
  \begin{equation*}
    \vec{v}= (b2^{-\delta},0,\ldots,0),
  \end{equation*}
  where $b=0\bmod q$ and $\|\vec{v}\|_{\infty}$ is the maximum of the
  absolutes values of vector $\vec{v}$.
\end{lemma}
Notice that this requires that $\beta_i$ are realizations of random independent
variables in $\F[q]$ and, unfortunately, as it is mentioned in~\cite{Nguyen02},
this is not necessary the case. However, if $\beta_i$ are sufficiently
well-distributed, then the situation remains the same.

In order to keep the paper self-contained, we
recall a way to measure well-distribution through the concept of
\textit{discrepancy}.
\begin{definition}
  Let $\Gamma$ be a multiset of  $N$ points contained  in the
  real interval $[0,1)$, then the discrepancy of the set
  is defined as
  \begin{equation*}
    \Discrepancy{\Gamma} = \sup_{\cB \subseteq [0,1)}
  \left|\frac{T_\Gamma(\cB)} {N} - |\cB|\right|,
  \end{equation*}
where $T_\Gamma(\cB)$ is the number of points of  $\Gamma$
inside the interval
\begin{equation*}
  \cB = [\cB^1, \cB^2) \subseteq [0,1)
\end{equation*}
of volume $|\cB| = \cB^2-\cB^1$ and the supremum is taken over
all such boxes.
\end{definition}
From the definition, it is easy to see that the discrepancy is a number between
$0$ and $1$. The closer the value is to $0$, more uniformly is distributed in
the unit interval. For more information about discrepancy, see~\cite{DrTi}.
We also need to introduce the following definition.
\begin{definition}
  A set $\cT$ of integers is $\Delta$-homogeneously distributed modulo
  $q$ if for any  integer $b$ coprime with $q$ the discrepancy of the
  set,
  \begin{equation*}
    \left \{ \frac{bt\bmod q}{q}\ \mid\
      t\in\cT
  \right \}
  \end{equation*}
  is at most $\Delta.$
\end{definition}
We now state the following lemma. from~\cite{Nguyen02}.
\begin{lemma}[\cite{Nguyen02}]
\label{lemma:2Nguyen}
  Let $a\in\F$ be different from zero. Choose
  integers $\beta_2,\ldots, \beta_n$ uniformly and independently at
  \textit{ random from $\cT$,
    which is $\Delta-$homogeneously distributed modulo $q$. }
  Then with probability $P\ge 1- q(2^{-\delta}+\Delta)^{n-1}$
  all vectors $\vec{v}\in\cL$ such that $\|\vec{v}\|_{\infty}\le
  2^{M-\delta}$ are of the form
  \begin{equation*}
    \vec{v}= (b2^{-\delta},0,\ldots,0),
  \end{equation*}
  where $b=0\bmod q$ and $\|\vec{v}\|_{\infty}$ is the maximum of the
  absolutes values of vector $\vec{v}$.
\end{lemma}
To show  the limits of the attack, it is
necessary to show that $\beta_2,\ldots, \beta_n$ defined in
Equation~\eqref{eq:linear_equations} are taken from a set
$\Delta-$homogeneously distributed. For this reason,
we improve~\cite[lemma 10]{Nguyen02}, which could be of independent interest
and show that the following set,
\begin{equation*}
  \Gamma=\left \{ s^{-1}m\bmod q\ \mid\
    \text{ where }s,\ m \text{ are defined in Equation~\eqref{eq:definition_rs}}
  \right \},
\end{equation*}
is $q^{-d}$-homogeneously distributed.
\begin{lemma}
\label{lemma:3Nguyen}
  Fixed a real number $1/2>\varepsilon >0$, then
  for any sufficiently big $p$, there
  exists $d>0$ such that for any $g\in\F[p]$ of multiplicative
  order $q \ge p^{\varepsilon}$, the set
  $\Gamma$ is  $q^{-d}$-homogeneously distributed provided that the
  hash function verifies Assumption~\ref{ref:assumption}.
\end{lemma}
The proof of this lemma will be given in Section~\ref{sec:exponential}.
By lemma~\ref{lemma:2Nguyen}, a discrepancy bound
for the set
\begin{equation*}\label{def:BarGamma}
  \overline{\Gamma}=\left \{ \frac{ b s^{-1}m\bmod q}{q}\ \mid\
    \text{ where } s,\ m \text{ are defined in Equation~\eqref{eq:definition_rs}}
  \right \},
\end{equation*}
for any $b$ coprime with $q$ gives a bound for the probability of $\cL$ having a
sufficiently short vector. Lemma~\ref{lemma:3Nguyen} alone is not sufficient to
measure the limits of the attack. Also, it is important to note that to find the
closest vector in a lattice to a given target is an \textit{NP-complete} problem
if the dimension of the lattice is a parameter. The attacker relies on
algorithms that provide only approximations for the closest vector in a lattice
when the dimension is large. In particular, he can use a combination of
Schnorr's modification~\cite{schnorr} of the LLL algorithm with the result of
Kannan to approximate the CVP~\cite{kannan}. We thus have the following result.
\begin{lemma}
  \label{lemma:schnorr}
  There exists a polynomial time algorithm which, given an
  $n$-dimensional full rank lattice $\cL$ and a vector $\vec{r}$, finds a
  a vector $\vec{v}\in\cL$ satisfying the inequality,
  \begin{equation*}
    \|\vec{r}-\vec{v}\|\le 2^{O(n\log^2\log(n)/\log(n))}
    \min\{\|\vec{r}-\vec{h}\|\ \mid\ \vec{h}\in\cL\},
  \end{equation*}
  where the implied constants are absolute.
\end{lemma}
This lemma shows that we must also consider the cases where the vector
found is not so short, and this is the reason that proving results for
$\delta$ small is difficult.

There is also an added difficulty, coming from $\delta$. Not all the bits of the
ephemeral key $k_i$ are taken randomly and independently, indeed only $N-\delta$
are taken at random and the rest are fixed. The case of many blocks of shared
bits is difficult because there are several blocks of bits which are fixed.
However, in this case, one can prove a bound for the discrepancy of the set
$\overline{\Gamma}$.
\footnote{We notice that $\beta_i$ are elements of the set
  $\overline{\Gamma}$
plus  $s_1^{-1}m_1$ and then reduced modulo $q$. But this does not
change the value of the discrepancy.}
We cite the following lemma without proof because its independent
interest and mention
that this follows the same lines as the previous result.
\begin{lemma}
\label{lemma:blocks}
  Fixed a real number $1/2>\varepsilon >0$, then for any sufficiently big $p$,
  there exists $d>0$ such that for any $g\in\F[p]$ of multiplicative order $q\ge
  p^{\varepsilon}$, $\Gamma$ is a $q^{-d}(\log q)^{\ell}$-homogeneously
  distributed \textit{when $k$ is taken with $\ell$ blocks of bits fixed }
  provided that the hash function verifies Assumption~\ref{ref:assumption}.
\end{lemma}
As explained above, we focus on the case $\ell=1$ and thus we will prove this
lemma in the next section for the case where $\ell=1$.

This result gives useful information whenever the discrepancy of the set
$\overline{\Gamma}$ is smaller than $2^{-\sqrt{\log q}}$. We see that if $\ell$
is fixed and $q$ and $\delta$ are big enough, then the attack has a high
probability of success.

\section{Main results}
\label{sec:exponential}
\subsection{Exponential Sums and Discrepancy}
In this section, we study the discrepancy of the set
$\overline{\Gamma}$ (see the definition on page
\pageref{def:BarGamma}) in the unit interval,
Typically the bounds on the discrepancy of a
sequence  are derived from bounds of exponential sums
with elements of this set. The relation is made explicit in
the celebrated {\it Koksma--Sz\"usz inequality\/} which we  present in the
following form.
\begin{lemma}[Corollary~3.11, \cite{Nied2}]
\label{lem:Kok-Szu}
Let  
$\Xi$ be a set of $N$ points in the range $[-q/2,\ldots, q/2]$
 such that  there exits a real
number $B$  with the property
\begin{equation*}
\left | \sum_{x\in\Xi} \exp \( 2 \pi i \frac{ u x }{q} \)
\right |\le B,
\end{equation*}
for any integer  $u$ with $u\neq 0 $ and $-q/2<u\le q/2$.
Then,
the discrepancy $\Discrepancy{\overline{\Xi}}$ where,
\begin{equation*}
  \overline{\Xi}=\left \{\ \frac{x}{q}\ |\ x\in\Xi\ \right\},
\end{equation*}
satisfies
\begin{equation*}
\Discrepancy{\Xi} \ll \frac{B\log q}{N},
\end{equation*}
where the implied constant is absolute.
\end{lemma}

For a positive integer $r$ we denote
\begin{equation*}
  \e[r](z) = \exp(2 \pi i z/r),\text{ where z is an integer}.
\end{equation*}

Notice that for a prime $r=q$, the function $\e(z)$ is an additive character
of $\F$. Exponential sums are well studied and used
extensively in number theory, uniform distribution theory and many
other areas because of their applications. In the following lemmas, we
outline several known properties.
\begin{lemma}[Exercise 11.a,Chapter 3, \cite{vinogradov}]
\label{lemma:linear}
  Then, for any set $\cK\subset\F$ and $k\in\F$, the formula
  \begin{equation*}
    \sum_{u\in\F}\sum_{k'\in\cK}\e[q]\(u(k-k')\)=
    \begin{cases}
      0 & \text{ if }k\not\in\cK,\\
      q & \text{ otherwise,}
    \end{cases}
  \end{equation*}
  holds.
\end{lemma}
\begin{lemma}[Exercise 11.c,Chapter 3, \cite{vinogradov}]
  \label{lemma:incomplete_linear}
  For any $1\le h\le q$ and $u\in\F$, $u\neq 0$, the following inequality,
  \begin{equation*}
  \frac{1}{q}\sum_{x\in\F}\left |   \sum_{y=1}^{h}\e(uxy)\right |\ll
  \log q,
\end{equation*}
holds, where the implicit constant is absolute.
\end{lemma}
We will need the following version of the \textit{Weil bound}.
\begin{lemma}[\cite{morenomoreno}]
\label{lem:Weil}
Let $F/G$ be a non-constant univariate rational function over $\F$
and let $v$ be the number of distinct roots of the polynomial $G$ in
the algebraic closure of $\F$. Then
\begin{equation*}
  \left |\sum_{x\in\F}\hskip-15 pt {\phantom{{\Sigma^2}}}^*
  \e\left (\frac{F(x)}{G(x)}\right
) \right|\le (\max (\deg F,\deg G)+ v^*-2)\ q^{1/2}+\rho,
\end{equation*}
where $\Sigma^*$ indicates that the poles of $F/G$ are excluded from
the summation,
$v^*=v$ and $\rho=1$ if $\deg F\le \deg G$, otherwise
$v^*=v+1$ and $\rho=0$.
\end{lemma}
In order to prove the main result of this paper, we will need to study the
number of solutions of the left part of Eq.~\eqref{eq:definition_rs} when $k$
has some fixed bits. Nguyen and Shparlinski~\cite[lemma 8]{Nguyen02} proved a
similar result but we prove a stronger bound using the following result.
\begin{lemma}[Theorem 4.1, \cite{garaev2010sums}]
  \label{lemma:Garaev}
  Let $3\le m\le 1.44\log\log p$ be a positive integer, and $c>0$ an arbitrary
  fixed constant. Suppose that $X_1,\ldots, X_m$ are subsets of $\F[p]$
  not containing $0$ and satisfying the condition
  \begin{equation*}
    |X_1|\cdot |X_2|\cdot \(|X_3|\cdots |X_m|\)^{1/81}>p^{1+c}.
  \end{equation*}
  Then,
  \begin{equation*}
    \left | \sum_{x_1\in X_1}\cdots \sum_{x_m\in X_m} \e[p](x_1\cdots x_m)\right |
    \le |X_1|\cdots |X_m| p^{-0.45c/2^m}.
  \end{equation*}
\end{lemma}
The following result is a particular case of the one given
in~\cite[Corollary 4.1]{garaev2010sums},
but we prove here an explicit version for this case.
\begin{lemma}
  \label{lemma:multiplicative_group_bound}
  Fixed a real number $1/2>\varepsilon >0$, then
  for any sufficiently big $p$ and  $g\in\F[p]$ of multiplicative order
  $q \ge p^{\varepsilon}$, the following bound,
  \begin{equation*}
    \max_{\gcd(c,p)=1}\left | \sum_{k=1}^{q}\e[p](cg^{k}) \right |\le
    q^{1 - 2^{145-82/\varepsilon}},
  \end{equation*}
  holds.
\end{lemma}
\begin{proof}
  The proof is just the application of lemma~\ref{lemma:Garaev}.
  Fix the value of $c$ to $1/81$, and select an integer $m$ satisfying,
  \begin{equation*}
    q^{2+ (m-2)/81}>p^{82/81}\implies \varepsilon(m+160)> 82,
  \end{equation*}
  where the inequality on the right has been obtained by substituting
  $q=p^{\varepsilon}$ and taking logarithms in the equality on the right.

  Now, considering $X_1 =  X_2= \cdots = X_{m-1} = \{g^{k}\bmod p\;|\; k =
  1,\ldots, q\}$, $X_m = \{c x_1\;|\; x_1\in X_1\}$ and
  lemma~\ref{lemma:Garaev}, gives
  \begin{equation*}
    \left | \sum_{k=1}^{q}\e[p](cg^{k}) \right |^{m}\le q^{m} p ^{-2^{-m}/180}\implies
    \left | \sum_{k=1}^{q}\e[p](cg^{k}) \right |\le q p ^{-2^{-m}/(180m)}.
  \end{equation*}
  Now, select $m$ the minimum integer satisfying $(m+160)\varepsilon>82$. If $p$
  satisfies $m\le 1.44\log\log p$, i. e. it is sufficiently big, substituting
  the minimum value of $m$ and $q\ge p^{\varepsilon}$ give the result.
\end{proof}
The next result is a generalization of a result by Shparlinski and
Nguyen~\cite[lemma 8]{Nguyen02}. In this result, we prove an asymptotic bound
for the discrepancy for the elements of the multiplicative group generated by $g$
for sufficiently big $p$.
\begin{lemma}
\label{lemma:8Nguyen}
  Fixed a real number $1/2>\varepsilon >0$, then
  for any sufficiently big $p$ and  $g\in\F[p]$ of multiplicative order
  $q = p^{\varepsilon}$,  the number of
  solutions of the following equation,
  \begin{equation*}
    g^{k}\mod p\mod q = \rho,\quad \rho\in\F,
  \end{equation*}
  where $1\le k\le  q$ is  $O (q^{1 - 2^{145-82/\varepsilon}}\log p)
  $, where the implied constant is absolute.
\end{lemma}
\begin{proof}
  Defining $L =\lceil p/q \rceil$, it is only necessary to bound the number of
  solutions of
  \begin{equation*}
    (g^{k} - q x )\mod p = \rho, \quad 0\le x\le L,
  \end{equation*}
  and where $\rho\in\F[p]$ is fixed. By lemma~\ref{lemma:linear}, the number of
  solutions is bounded by
  \begin{equation*}
    \left | \frac{1}{p}\sum_{c=1}^{p}\sum_{x=0}^{L}\sum_{k=1}^{q}
      \e[p](c(g^{k} - q x -\rho))\right | =
    \left | \frac{1}{p}\sum_{c=1}^{p} \e[p](-c\rho)
      \sum_{k=1}^{q}\e[p](cg^k)\sum_{x=0}^{L}\e[p](-qcx)\right |.
  \end{equation*}
Now, we bound this sum using lemmas~\ref{lemma:incomplete_linear}
and~\ref{lemma:multiplicative_group_bound} so
\begin{multline*}
   \left | \frac{1}{p}\sum_{c=1}^{p} \e[p](-c\rho)
      \sum_{k=1}^{q}\e[p](cg^k)\sum_{x=0}^{L}\e[p](-qcx)\right |\\
    \le
    2 + \frac{1}{p} \sum_{c =1}^{p-1}\left |\sum_{k=1}^{q}\e[p](cg^k)\right |
    \left | \sum_{x=0}^{L}\e[p](-qcx) \right |
    \ll q^{1 - 2^{145-82/\varepsilon}}\log p.
\end{multline*}
This finishes the proof.
\end{proof}
Now, we will give an upper bound of the
following exponential
\begin{equation*}
  \sum_{m\in\cH(\cM)}\sum_{k\in\F}\hskip-15 pt {\phantom{{\Sigma^2}}}^*
  \e[q](c(\beta(k,m))+uk),
\end{equation*}
where $\beta(k,m)$ is defined as,
\begin{equation*}
  \beta(k,m) := 2^{-t}(s^{-1}r)\bmod q,
\end{equation*}
and $s,r$ are defined in~\eqref{eq:definition_rs}. The symbol $\Sigma^*$
indicates that the poles are excluded from summation.
\begin{lemma}
\label{lemma:main}
  Fixed a real number $1/2>\varepsilon >0$, then
  for any sufficiently big $p$ and any $g\in\F[p]$ of multiplicative
  order $q = p^{\varepsilon}$, the bound
\begin{equation*}
\max_{\gcd(c,q)=1}
\left |\sum_{m\in\cH(\cM)}\sum_{k\in\F[q]}\hskip-15 pt
  {\phantom{{\Sigma^2}}}^*\e[q](c(\beta(k,m))+uk)\right
|\ll  W^{1/2}q^{3/2 - 2^{145-82/\varepsilon}}\log^2 p,
\end{equation*}
holds, where the constant is absolute.
\end{lemma}
\begin{proof}
Taking any integer $c$ coprime with $q$ and
 calling $\sigma$ the
  value of the exponential sum, we have
  \begin{equation*}
    \sigma =\left|
        \sum_{m\in\cH(\cM)}
        \sum_{k\in\F[q]}\hskip-15 pt {\phantom{{\Sigma^2}}}^*
      \e[q](c(\beta(k,m))+uk)
      \right | \le
        \sum_{m\in\cH(\cM)}\left|
          \sum_{k\in\F[q]}\hskip-15 pt {\phantom{{\Sigma^2}}}^*
          \e[q](c(\beta(k,m))+uk)
      \right |.
  \end{equation*}
For $\lambda\in\F$ we denote by $H(\lambda)$ the number
of $m\in\cH(\cM)$ with $m=\lambda$. We also define the integer
$c_0:=2^{-t}c\bmod q$. Then,
\begin{equation*}
  \sigma =\sum_{\lambda\in\F}H(\lambda)\left |
    \sum_{k\in\F[q]}\hskip-15 pt {\phantom{{\Sigma^2}}}^*
    \e[q]\(c_0\frac{kr(k)}{\lambda+ar(k)}+uk\)
  \right|,
\end{equation*}
where $a$ is the private key, the symbol $\Sigma^*$ indicates that the poles are
excluded from summation and $r(k):=g^{k}\bmod p\bmod q$

Now, we apply the Cauchy inequality,
\begin{equation}
\label{eq:after_cauchy}
  \sigma^2\le (\sum_{\lambda\in\F}H(\lambda)^2)
  \sum_{\lambda\in\F}\left |
    \sum_{k\in\F[q]}\hskip-15 pt {\phantom{{\Sigma^2}}}^*
    \e[q]\(c_0\frac{kr(k)}{\lambda+ar(k)}+uk\)
  \right |^2.
\end{equation}
We remark that,
\begin{equation}
\label{eq:remark}
  \sum_{\lambda\in\F}H(\lambda)^2=W.
\end{equation}
We can operate with
the other term in the right side of Equation~\eqref{eq:after_cauchy}.
\begin{multline*}
  \sum_{\lambda\in\F}\left |
    \sum_{k\in\F[q]}\hskip-15 pt {\phantom{{\Sigma^2}}}^*
    \e[q]\(c_0\frac{kr(k)}{\lambda+ar(k)}+uk\)
  \right |^2\le\\
  \sum_{\lambda\in\F}
  \sum_{k_1,k_2\in\F[q]}\hskip-15 pt {\phantom{{\Sigma^2}}}^*
  \e[q]\(c_0\(\frac{k_1r(k_1)}{\lambda+ar(k_1)}
  -\frac{k_2r(k_2)}{\lambda+ar(k_2)}\)
  +u(k_1-k_2)
  \)  =
  \\
  \sum_{k_1,k_2\in\F[q]}
  \sum_{\lambda\in\F}\hskip-15 pt {\phantom{{\Sigma^2}}}^*
  \e[q]\(c_0\(\frac{k_1r(k_1)}{\lambda+ar(k_1)}
  -\frac{k_2r(k_2)}{\lambda+ar(k_2)}\)
  +u(k_1-k_2) \) .
\end{multline*}
We write the inner sum in the following way:
\begin{equation*}
  \sum_{\lambda\in\F}\hskip-15 pt {\phantom{{\Sigma^2}}}^*
  \e[q]\(c_0\(\frac{k_1r(k_1)}{\lambda+ar(k_1)}
  -\frac{k_2r(k_2)}{\lambda+ar(k_2)}\)
  +u(k_1-k_2) \) =
  \sum_{\lambda\in\F}\hskip-15 pt {\phantom{{\Sigma^2}}}^*\e[q](F_{k_1,k_2}(\lambda))
\end{equation*}
where
\begin{equation*}
  F_{k_1,k_2}(X):=
\(c_0\(\frac{k_1r(k_1)}{X+ar(k_1)}
  -\frac{k_2r(k_2)}{X+ar(k_2)}\)
  +u(k_1-k_2) \).
\end{equation*}
Notice that $F_{k_1,k_2}(X)$ is a rational function  in $X$
when $k_1$ and $k_2$ are fixed.
The function is not constant if $r(k_1)\neq r(k_2)$ because then
$F_{k_1,k_2}$ has two different poles.
If $r(k_1)=r(k_2)$, the sum is constant  only in two
cases:
either $k_1=k_2$ or $r(k_1)=r(k_2)=0$. By lemma~\ref{lemma:8Nguyen},
we see that the number of such pairs is $O(q^{2 - 2^{146-82/\varepsilon}}\log^2 p+q)$.
In other case, it is easy to see that $F_{k_1,k_2}(X)$
is not a constant function so it is possible to apply
lemma~\ref{lem:Weil}. This gives,
\begin{equation*}
    \sum_{\lambda\in\F}\left |
    \sum_{k\in\F}\hskip-15 pt {\phantom{{\Sigma^2}}}^*
    \e[q]\(c_0\frac{kr(k)}{\lambda+ar(k)}+uk\)
  \right |^2 \ll 
  \(q^{3 - 2^{146-82/\varepsilon}}\log^2 p \).
\end{equation*}
Substituting this
estimate in Equation~\eqref{eq:after_cauchy} with
Equation~\eqref{eq:remark}, we get the result.
\end{proof}

\subsection{Proof of lemma~\ref{lemma:blocks} for $\ell=1$}

lemma~\ref{lem:Kok-Szu}  shows the relationship between
bounds on exponential sums and bounds on the discrepancy.
So, our goal is to find bounds of the following exponential sum:
\begin{equation*}
  \left |
    \sum_{m\in\cH(\cM)}\sum_{k\in\cK}
    \e[q] \(c\beta(k,m)\)
  \right |=
  \frac{1}{q}\left | \sum_{k\in\F}
    \sum_{m\in\cH(\cM)}\hskip-15 pt {\phantom{{\Sigma^2}}}^*
    \e[q] \(c\beta(k,m)\)
    \sum_{u\in\F}\sum_{k'\in\cK}\e[q](u(k-k'))
  \right |,
\end{equation*}
where $\cK$ is the set of integers defined by
Equation~\eqref{definition:key}.

Notice that if $k$ does not meet the requisites, i. e. it does not has
the correct bits fixed, the inner sum is equal to zero. Otherwise, the
inner sum is equal to one. Doing the following transformations,
\begin{multline*}
   \left | \sum_{k\in\F}
    \(\frac{1}{q}\sum_{u\in\F}\sum_{k'\in\cK}\e(u(k'-k))\)
    \sum_{m\in\cH(\cM)} \e[q] \(c\beta(k,m)\)
\right |=\\
\left | \frac{1}{q}\sum_{u\in\F}\sum_{k\in\F}
  \(\sum_{k'\in\cK}\e(uk'))\)
  \sum_{m\in\cH(\cM)} \e[q]
  \(c\beta(k,m)-uk\)
\right |\le
\\
\frac{1}{q}\sum_{u\in\F}
\left | \sum_{k\in\F}\sum_{m\in\cH(\cM)} \e[q]
  \(c\beta(k,m)-uk\) \right |\left |
  \sum_{k'\in\cK}\e(uk'))
\right |.
\end{multline*}
By lemma~\ref{lemma:main}, we have that
\begin{equation*}
\max_{\gcd(c,q)=1}  \left | \sum_{m\in\cH(\cM)} \sum_{k\in\F}
   \hskip-15 pt {\phantom{{\Sigma^2}}}^*
  \e[q]
  \(c\beta(k,m)-uk\) \right |\ll
W^{1/2}q^{3/2 - 2^{145-82/\varepsilon}}\log^2 p.
\end{equation*}
Recalling lemma~\ref{lemma:incomplete_linear},
\begin{multline*}
  \left | \sum_{k\in\cK}
    \sum_{m\in\cH(\cM)}\hskip-15 pt {\phantom{{\Sigma^2}}}^*
    \e[q] \(c\beta(k,m)\)
  \right |\le\\
  \frac{1}{q}\sum_{u=1}^{q}
  \left | \sum_{k=1}^{q}\sum_{m\in\cH(\cM)} \e[q]
    \(c\beta(k,m)-uk\) \right |\left |
    \sum_{k'\in\cK}\e(uk'))
  \right | \ll
  \\
  W^{1/2}q^{1/2 - 2^{145-82/\varepsilon}}\log^2 p \sum_{u\in\F}\left |
    \sum_{k'\in\cK}\e(uk')\right |\ll
  W^{1/2}q^{3/2 - 2^{145-82/\varepsilon}}\log^3 p.
\end{multline*}
The above bound for the exponential sum and lemma~\ref{lem:Kok-Szu}
show that $\Gamma$ is a
$2^{-\log^{1/2}q}$-homogeneously distributed modulo $q$ provided that
\begin{equation*}
  W\le \frac{|\cM|}{q^{1 - 2^{146-82/\varepsilon}}\log^6 p}.
\end{equation*}

\subsection{Proof of theorem~\ref{theorem:main} and some comments}
\label{sec:main}

Now, we are ready to prove the main result.

  Suppose that the attacker has obtained the following messages with
  their corresponding signatures,
  \begin{equation*}
    (m_1,s_1,r_1),\ldots, (m_n,s_n,r_n).
  \end{equation*}
  Using this information, the attacker builds lattice~$\cL$ using the
  rows of the matrix defined in~\eqref{eq:matrix} and also
  vector~$\vec{t}$ defined in~\eqref{eq:twovectors}.  The attacker can
  find a closest vector in $\cL$ to $\vec{t}$
  and suppose that the second component of this vector is
  $\gamma_2$.  Let $\vec{h}\in\cL$ be the solution found to the
  closest vector problem, so the norm of
  $\vec{h}-\vec{t}$ 
  satisfies,
  \begin{equation}
    \label{eq:bound_uh}
    \|\vec{u}-\vec{h}\|\le \|\vec{u}-\vec{t}\|+\|\vec{t}-\vec{h}\|\le
    \sqrt{n}2^{M-\delta+1} = 2^{M-\delta+\log n+1}.
  \end{equation}
  The attack success if any vector in the lattice with norm less than
  $2^{M-\delta+\log n+1}$ has a zero in the second coordinate.
  By lemma~\ref{lemma:2Nguyen}, the probability of success is greater
  than $1- q(2^{-\delta+\log n+1}+\Delta)^{n-1}$, if $\Gamma$ is a
  $\Delta$-homogeneously distributed. Lemma~\ref{lemma:blocks} implies
  that it is possible to take $\Delta=q^{-d}$ and this finish the
  proof.

We want to mention that if the dimension is greater than $100$, only
approximation algorithms are practical. In those cases, it is
necessary to multiply in  the right hand side of
equation~\eqref{eq:bound_uh} by the factor
appearing in lemma~\ref{lemma:schnorr}. This gives a lower bound in
the probability of~$1- q(2^{-\delta+O(n\log^2\log n)/\log n}+\Delta)^{n-1}.$

\section{Experimental results}
\label{sec:experiments}

We have empirically tested the performance of the attack.

In the first parameters set, the bit size of
$p$ is $1024$ and the bit size of $q$ is $160$. In the second set, the bit
size of $p$ is $4096$ and the bit size of $q$ is $250$. For the hash
function, we have chosen SHA1, because it was widely used in DSA.

We note that the experimental results are better than what we expect from
Theorem~\ref{thm:general}.

The reason is that the lower bound, $d\ge 2^{145-82/\varepsilon}-\log\log p/\log q$,
is very pessimistic. Indeed, for $d\approx 0.5$, we have made the following calculations
in Table~\ref{tab:comparation}. The calculation for the theoretical value of $n$ is
finding the minimum value such that,
$$
1-q(2^{-\delta+\log n+1}+ q^{-d}(\log q)^{\ell})^{n-1}>0.
$$
From the empirical results, we make the following conjecture.
\begin{conjecture}
  Assuming that $\delta\ge \log^{1/2} M$, then given $n$ messages
  with $ n\ge M/\log^{-1/2} (M),$ the probability of
  success is greater than $1- 1/n$.
\end{conjecture}

\begin{table}
  \centering
  \begin{tabular}{|c|c|c|}
    \hline
    Known bits& Mean value of $n$  in simulations&Value of $n$  by Theorem~\ref{thm:general}\\\hline
    30& 7.0& 7\\\hline
    28& 7.4& 8\\\hline
    26& 8.0& 8\\\hline
    24& 8.5& 9\\\hline
    22& 9.0& 10\\\hline
    20& 9.8& 11\\\hline
    18& 11.0& 12\\\hline
    16& 12.5& 14\\\hline
    14& 14.3& 17\\\hline
    12& 17.1& 21\\\hline
    10& 20.7& 30\\\hline
    8& 27.3& 54\\\hline
    6& 50& -\\\hline
  \end{tabular}
  \caption[Theoretical vs. Simulations]{Comparation of theoretical values and mean of computer simulations when
    employing  $p$ of 1024 bits and $q$ of 160 bits. 
    The mean value of messages needed has been simulated adding
  another message until the attack is successful.}
  \label{tab:comparation}
\end{table}
Figure~\ref{fig:1} show $M-\delta$ in the abscissa against the minimum number of signatures required to
recover the ephemeral keys. For each experiment, we have selected randomly the value $k',\ k''$  and $\tilde{k_i}$ for $i=1,\ldots,n$ and
repeated each experiment $10$ times.
\begin{figure}[t]
  \centering\includegraphics[width=100mm]{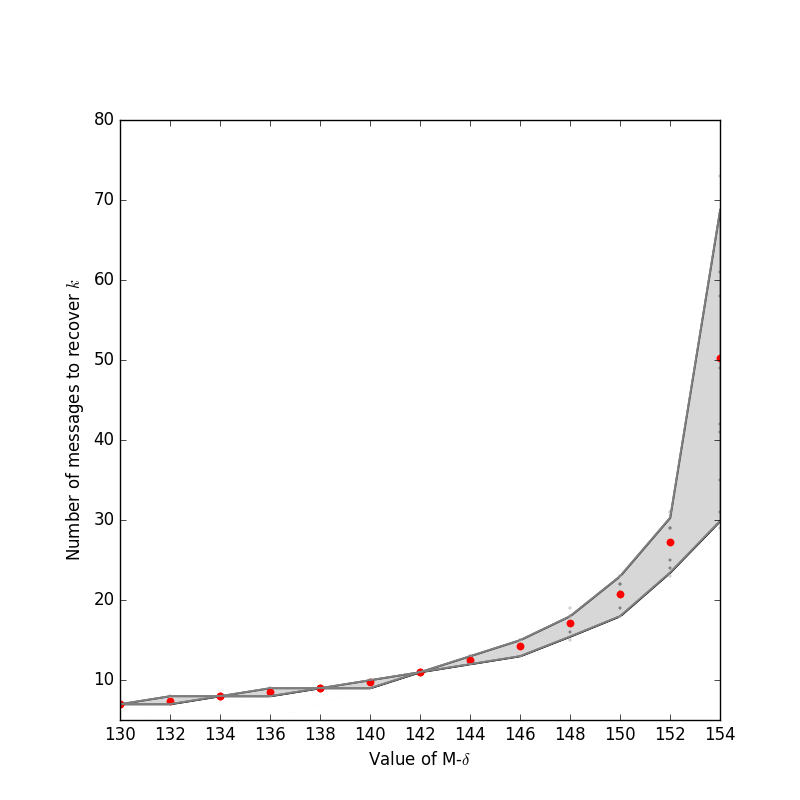}
    \centering\includegraphics[width=100mm]{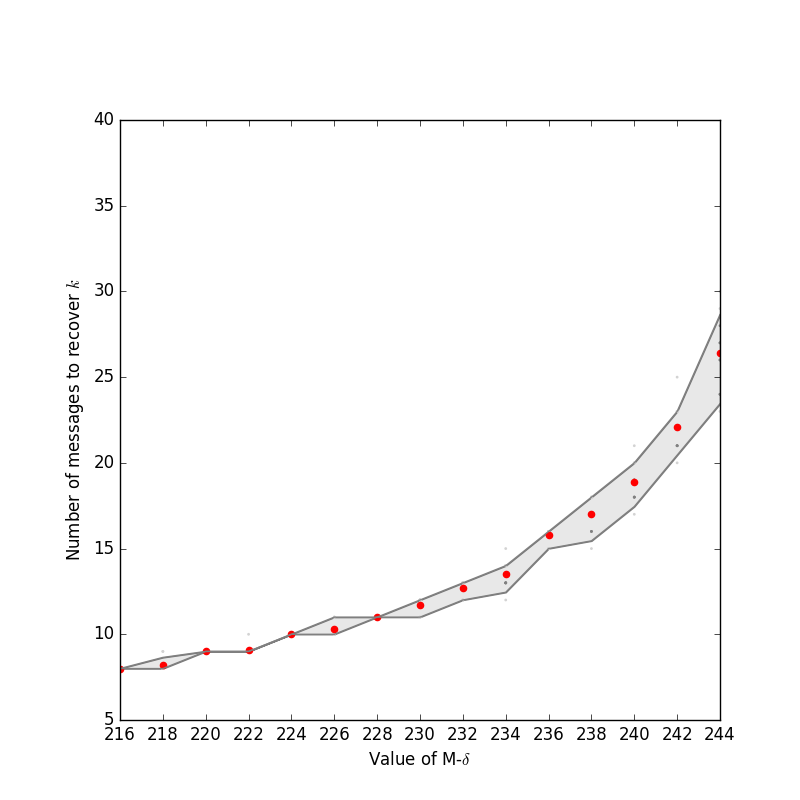}
  \caption{Experimental results of the algorithm. The first plot employs the smaller parameter set, with $p$ of 1024 bits and $q$ of 160 bits while the second plot employs the parameter set with $p$ of 4096 bits and $q$ of 250 bits.
    The red dots represent the average of
    the number of message signatures needed to recover the ephemeral keys.
    The grey areas contain the 90
    percentile of the number of message signatures needed}
  \label{fig:1}
\end{figure}

\section*{Acknowledgments}
We want to thank Igor Shparlinski for his time,
ideas and comments during the development of the paper.
The research of the first author was supported by the Ministerio de Economia y Competitividad research project MTM2014-55421-P.

\end{document}